\documentclass{article}

\usepackage{booktabs} 
\usepackage{amsfonts,amssymb,amsmath,latexsym,ae,aecompl,bbm,algorithm,amsthm}
\usepackage[noend]{algpseudocode}
\usepackage{multicol}
\usepackage{caption}
\usepackage{fullpage}
\usepackage{graphicx}

\newtheorem{theorem}{Theorem}
\newtheorem{corollary}{Corollary}
\newtheorem{lemma}{Lemma}

\newtheorem{observation}{Observation}

\begin{document}
\title{Deterministic Dispersion of Mobile Robots in Dynamic Rings}

\author{Ankush Agarwalla\thanks{Department of Computer Science \& Engineering, National Institute of Technology Rourkela, Rourkela, India. ankush.agarwalla7@gmail.com. Work done during internship at IIT Madras.}
\and John Augustine\thanks{Department of Computer Science \& Engineering, Indian Institute of Technology Madras, Chennai, India. augustine@iitm.ac.in. Research supported in part by an Extra-Mural Research Grant (file number EMR/2016/003016) funded by the Science and Engineering Research Board, Department of Science and Technology, Government of India.}
\and William K. Moses Jr.\thanks{Department of Computer Science and Engineering, Indian Institute of Technology Madras, Chennai, India. wkmjr3@gmail.com. Research supported in part by an Extra-Mural Research Grant (file number EMR/2016/003016) funded by the Science and Engineering Research Board, Department of Science and Technology, Government of India.}
\and Madhav Sankar K.\thanks{Department of Computer Science \& Engineering, National Institute of Technology Tiruchirappalli, Tiruchirappalli, India. madhavsankar@gmail.com. Work done during internship at IIT Madras.}
\and Arvind Krishna Sridhar\thanks{Department of Computer Science \& Engineering, National Institute of Technology Tiruchirappalli, Tiruchirappalli, India. arvindkrishna1997@gmail.com. Work done during internship at IIT Madras.}
}

\date{}

\maketitle

\begin{abstract}

In this work, we study the problem of dispersion of mobile robots on dynamic rings. The problem of dispersion of $n$ robots on an $n$ node graph, introduced by Augustine and Moses Jr.~\cite{AM18}, requires robots to coordinate with each other and reach a configuration where exactly one robot is present on each node. This problem has real world applications and applies whenever we want to minimize the total cost of $n$ agents sharing $n$ resources, located at various places, subject to the constraint that the cost of an agent moving to a different resource is comparatively much smaller than the cost of multiple agents sharing a resource (e.g. smart electric cars sharing recharge stations). The study of this problem also provides indirect benefits to the study of scattering on graphs, the study of exploration by mobile robots, and the study of load balancing on graphs.

We solve the problem of dispersion in the presence of two types of dynamism in the underlying graph: (i) vertex permutation and (ii) 1-interval connectivity. We introduce the notion of vertex permutation dynamism and have it mean that for a given set of nodes, in every round, the adversary ensures a ring structure is maintained, but the connections between the nodes may change. We use the idea of 1-interval connectivity from Di Luna et al.~\cite{LDFS16}, where for a given ring, in each round, the adversary chooses at most one edge to remove.

We assume robots have full visibility and present asymptotically time optimal algorithms to achieve dispersion in the presence of both types of dynamism when robots have chirality. When robots do not have chirality, we present asymptotically time optimal algorithms to achieve dispersion subject to certain constraints. Finally, we provide impossibility results for dispersion when robots have no visibility.
\end{abstract}



\section{Introduction}
\label{sec:intro}

\subsection{Background and Motivation}
The paradigm of using mobile robots to study distributed problems in a network is an interesting and much studied area. Each robot represents a computational machine with the ability to move, either in a geometric space or in a graph. Typically these robots must coordinate with each other to either achieve some specific configuration on the graph or else perform some sort of distributed calculation of a property of the graph (e.g., count the number of nodes of the graph).

The particular problem of dispersion was introduced in \cite{AM18} and asks $n$ robots on an $n$ node graph to coordinate and reach a configuration, as quickly as possible, such that exactly one robot is present on each node. It is generally applicable to any real world scenario where $n$ agents must coordinate and share $n$ resources, located at various places, where we want to minimize the total cost of solving the problem subject to the constraint that the cost of agents moving around on the graph is dwarfed by the cost of having more than one agent share a given resource. For example, consider smart electric cars coordinating amongst themselves to find recharge stations. It may take 30 minutes to drive to reach any given station, but it takes hours to recharge a given car. We want all cars to recharge as soon as possible. Thus, it is in the interests of the car manufacturer to design an algorithm that would allow cars to coordinate with each other and end up such that at most one car uses any given recharge center, when possible.

In addition to its immediate real world applications, the study of dispersion benefits three other research areas: (i) the study of scattering or uniform deployment of agents on a graph, (ii) the study of exploration of mobile robots on a graph, and (iii) the study of load balancing on a graph. 

For a given graph with $n$ nodes, the problem of scattering \cite{BFMS11,EB11,SMOKM16} of $k \leq n$ robots requires robots to coordinate and move such that they end up uniformly deployed over the graph. When $k=n$, the problem of scattering is exactly that of dispersion on the same graph.

For a given graph with $n$ nodes, $n$ robot collaborative exploration \cite{DMKJSC06,FPGLKDPA06,DMLJSC07,BPMFGAXJ11,BPVIXN14,HYKNLSTS14,DYMFNAKNSA16} asks that these robots coordinate amongst each other and collectively visit every node of the graph at least once. Notice that any solution to dispersion immediately applies to $n$ robot collaborative exploration for the same assumptions and model parameters. Thus any solutions for dispersion can be directly applied to collaborative exploration.

Load balancing on graphs \cite{BV86,PV89,C89,SS94,MGS98} asks that $n$ resources be equally split amongst the $n$ nodes on the graph in as few rounds as possible. Typically, the nodes control where different resources go and there are usually constraints on how much load may pass through a given edge in a round. If we consider robots as the resources, notice that the problem is very similar to dispersion. It is our hope that further study of dispersion will eventually lead to a mechanism that would allow cross-use of techniques and tools between the areas of mobile robots and load balancing.

The introduction of dynamism into the network makes the problem very interesting, as modern day networks change rapidly. Mobile robots themselves as physical robots may have to perform their tasks in physically changing environments. With the relationship dispersion has to load balancing, it may be possible to uncover useful strategies to deal with balancing resources when processors go down and are replaced by others or the connections between different processors change.

\subsection{Related Work}
The problem of dispersion using mobile robots was first studied in \cite{AM18}, where the authors studied the trade-off between time to dispersion and memory required by each robot for different types of static graphs.

A related problem is that of scattering or uniform deployment, where $k \leq n$ robots have to uniformly deploy themselves in a given network with $n$ nodes. \cite{BFMS11} looked into this problem on grids while \cite{EB11,SMOKM16} looked into this problem on static rings. In \cite{EB11}, they present various possibility and impossibility results in a synchronous system for oblivious robots with limited visibility and which don't know the size of the ring or total number of robots. However, for their algorithms they require that no two robots start on the same node, which in the case of $n$ robot scattering, implies that scattering is trivial. In \cite{SMOKM16}, they present various possibility and impossibility results in an asynchronous system.

Another closely related problem to dispersion is exploration, where one or more robots start at a given node and try to collectively visit all vertices of the graph as quickly as possible. Exploration of rings by mobile robots has been a long studied problem and there is much literature on it. In particular, several works have studied deterministic exploration of anonymous, unoriented rings by oblivious robots that relied on their view of the ring to make decisions \cite{LABMTB10,DALALLPF13,DALALLPF15,FPIDPASN13}. Exploration in an asynchronous setting was studied in \cite{LABMTB10,FPIDPASN13,DALALLPF15} while \cite{DALALLPF13} looked into the synchronous setting. 

In this paper, we prove impossibility results on myopic robots achieving dispersion on a dynamic ring. A robot is myopic when its visibility is restricted to a subset of nodes of the entire ring. \cite{DALALLPF13} provides positive and negative results for varied degrees of myopicness of robots for different levels of synchronicity of networks. Of particular interest, they provide synchronous deterministic algorithms for robots with no visibility to explore a static ring. For rings of size $3 \leq n \leq 6$, $k = n-1$ robots are required (note when $n=6$, only $4$ robots are required) and for rings of size $n \geq 7$, $5$ robots are required.

When $k$ robots all start at a given node and coordinate to explore a graph, it is referred to as $k$ robot collaborative graph exploration \cite{DMKJSC06,FPGLKDPA06,DMLJSC07,BPMFGAXJ11,BPVIXN14,HYKNLSTS14,DYMFNAKNSA16}. When $k=n$, any solution to dispersion also acts as a solution to $k$ robot exploration under the same conditions.

Exploration on dynamic rings has been investigated in \cite{IW13,LDFS16} under different variants of 1-interval connectivity dynamism. The notion of 1-interval connectivity dynamism was introduced in \cite{OW05} for complete graphs and requires an adversary to connect the nodes of the graph such that in every round there must exist a connected spanning subgraph of the nodes. This model was studied by \cite{OW05}, \cite{IW13}, and \cite{IKW14} for complete graphs, rings and cactuses respectively. \cite{LDFS16} studied the model where for a given ring of nodes, in each round an adversary chooses at most one edge of the ring and removes it for that round. This definition is restricted in that it does not allow nodes in the ring to have different neighbors in different rounds. In \cite{LDFS16}, for a fully synchronous system they proved that for an anonymous ring if size of ring is not known, then it is impossible to achieve exploration with termination. They also showed that if termination is not needed, exploration is possible under those conditions. If ring is anonymous and an upper bound on its size $N$ is known, then if robots have chirality, exploration with termination can be achieved in $3N$ time. Whereas without chirality, exploration with termination can be achieved in $5N$ time. In \cite{IW13}, they show that for a fully synchronous system, if an agent knows the entire graph every round, then exploration can be achieved in $2n - 2$ rounds. When the agent only knows its neighborhood each round, exploration is achieved in $n + n (\delta - 1) \pm \Theta(\delta)$ rounds, where $\delta$ is the length of the interval such that every edge of the graph is present at least once every $\delta$ rounds.

Exploration of graphs in presence of other types of dynamism has been studied extensively in the literature. Exploration of a dynamic graph where the graph is being modified by an oblivious adversary has been done in \cite{AKL08}. Exploration on temporal graphs has been studied in \cite{EHK15} where the edge set can change from step to step. Exploration of periodic time-varying graphs has been studied in \cite{FMS13} and \cite{FMS09}. 

\cite{KO11} serves as a good survey on the various models for dynamic networks, both random and adversarial, and algorithms for these models.

Load balancing is the uniform distribution of load among the nodes. If we think of the robots as load, then we can think of dispersion as load balancing on graphs with ``smart" load. We refer to each load as smart because it is the load that makes the decisions of where to move as opposed to the nodes. Divergence in the literature arises from the type of load being balanced (either discrete \cite{BV86,PV89} or continuous \cite{MGS98}) and the type of model employed (either diffusion or dimension exchange). In the case of diffusion \cite{C89,SS94,MGS98}, every node balances its load concurrently with all its neighbors in each round while in the dimension exchange model \cite{XL92}, each node balances load with at most one neighbor in each round. The work done in this paper is closer to the diffusion model of load balancing with discrete load. Load balancing on a ring in particular was looked at in \cite{BV86,PV89}.

\subsection{Our Results} 
We develop several deterministic algorithms for robots to achieve dispersion on a ring in the presence of different types of dynamism and the ability of robots to have chirality or not. All algorithms present have running time $O(n)$ rounds and are asymptotically time optimal. This is easy to see since if all robots are present on a given node in some configuration, it takes $\Omega(n)$ rounds for any robot to reach the node farthest away from the given node. We also provide impossibility results when robots have no visibility. To the best of our knowledge, this is the first work that analyzes dispersion in a dynamic setting.

We consider two types of dynamism: (i) vertex permutation and (ii) 1-interval connectivity. We introduce the notion of vertex permutation dynamism and use it to mean that for a given set of nodes, in every round, the adversary may change the order of the nodes while still maintaining a ring structure. We use the restricted definition of 1-interval connectivity from \cite{LDFS16} which says that for a given ring of nodes, in every round, the adversary can choose to remove at most one edge from the ring. Most of our main algorithms are valid in the presence of both vertex permutation and 1-interval connectivity. The combination of these two types of dynamism models the more general definition of 1-interval connectivity from \cite{OW05,KLO10} for the restricted setting of rings. 

Recall that an algorithm to solve dispersion on a ring also solves collaborative exploration and scattering on a ring under the same conditions and assumptions. Thus, this paper not only presents solutions to the problem of dispersion, but shows how multi-robot exploration and also scattering may be carried out in the presence of vertex permutation and 1-interval connectivity.

We first assume robots have chirality and full visibility and develop deterministic algorithms for increasing levels of dynamism. Algorithm \emph{VP-Chain} achieves dispersion on a ring in the presence of vertex permutation dynamism. We extend the ideas present in that algorithm and design algorithm \emph{VP-1-Interval-Chain} to achieve dispersion in the presence of both vertex permutation dynamism and 1-interval connectivity dynamism. 

When robots do not have chirality, we design algorithms to achieve dispersion under certain constraints in the presence of both vertex permutation and 1-interval connectivity dynamism. If all robots are initially located at the same node, then by running algorithm \emph{No-Chir-Preprocess} and then \emph{VP-1-Interval-Chain}, they can achieve dispersion. If the ring is an odd size ring, then robots running algorithm \emph{Achiral-Odd-VP-1-Interval-Chain} achieve dispersion. If the ring is of size $4$, then robots running algorithm \emph{Achiral-Even4-VP-1-Interval-Chain} achieve dispersion.

We subsequently show that deterministic dispersion is impossible to achieve in the face of either vertex permutation dynamism or 1-interval connectivity dynamism when robots have no visibility. We also provide corollaries that show that when the adversary has additional powers, dispersion with no visibility robots in the face of either type of dynamism is impossible even when algorithm designers have access to random bits.

\subsection{Organization of Paper}
The rest of this paper is organized as follows. Section~\ref{sec:prelims} provides definitions and technical preliminaries required for the rest of the paper. Section~\ref{sec:vp} provides our algorithm for robots with full visibility and chirality to achieve dispersion in the presence of only vertex permutation dynamism. Section~\ref{sec:vp-1-int} provides our algorithm for robots with full visibility and chirality to achieve dispersion in the presence of both vertex permutation and 1-interval connectivity dynamism. Section~\ref{sec:no-chir-vp-1-int} presents our results for dispersion in the presence of vertex permutation and 1-interval connectivity when robots have full visibility but do not have chirality. In Section~\ref{sec:no-vis-imp}, we present our impossibility results for achieving dispersion for both vertex permutation dynamism and 1-interval connectivity dynamism when robots have no visibility. Finally, in Section~\ref{sec:conc}, we present our conclusions and some open questions.


\section{Technical Preliminaries}
\label{sec:prelims}

\subsection{Structure of Ring}
We consider an $n$ node anonymous ring, i.e. each node does not have any label. The multiplicity of a node is the number of robots present on that node. We classify nodes in each round based on their multiplicities in that round. A \textbf{hole} is a node not containing any robots. A \textbf{singleton node} is a node containing exactly one robot. A \textbf{multinode} is a node containing two or more robots. A node containing one or more robots is said to be occupied. We define a \textbf{chain} as a sequence of nodes in the following order: a multinode followed by zero or more singleton nodes and finally a hole. When we only consider the order of nodes in the ring in clockwise order, we call any such chain that occurs a \textbf{clockwise chain}. Similarly, we define an \textbf{anticlockwise chain}. We define a \textbf{good chain} as a chain where the adversary has not removed an edge between any of the nodes of the chain. We define a \textbf{bad chain} as a chain where the adversary has removed an edge between two of the nodes of the chain. Examples of a clockwise good chain and a clockwise bad chain are illustrated in Figures~\ref{fig:clock-good-chain} and~\ref{fig:clock-bad-chain} respectively.

\begin{multicols}{2}
  	\begin{minipage}{\columnwidth}
  		\includegraphics[page=10,height=1.45in]{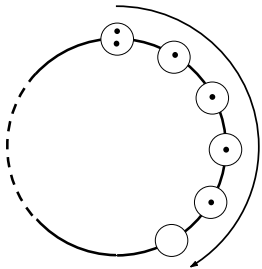}
  		\captionof{figure}{Example of a clockwise good chain}\label{fig:clock-good-chain}
  	\end{minipage}%
  	
  	\begin{minipage}{\columnwidth}
  		\centering
  		\includegraphics[page=11,height=1.45in]{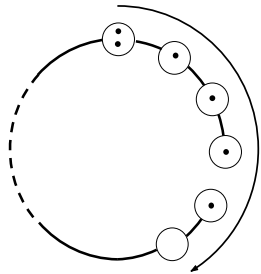}
  		\captionof{figure}{Example of a clockwise bad chain} \label{fig:clock-bad-chain}
  	\end{minipage}
\end{multicols} 

\subsection{Nature of Dynamism}
We consider two forms of dynamism: (i) vertex permutation and (ii) 1-interval connectivity. We introduce the notion of vertex permutation dynamism in this paper and define it as follows. In each round, the nodes of the ring are shuffled by an adversary such that their relative order may change. 

The notion of 1-interval connectivity dynamism was introduced in \cite{OW05} for complete graphs and requires an adversary to connect the nodes of the graph such that in every round there must exist a connected spanning subgraph of the nodes. \cite{OW05} studied it for complete graphs, \cite{IW13} studied it for the case of rings, and \cite{IKW14} looked into it for cactuses. Note that the connections between nodes need not remain from round to round. In this paper, we use a restricted version of 1-interval connectivity, proposed in \cite{LDFS16}, which requires that for a given ring of nodes, in each round an adversary chooses at most one edge of the ring and removes it for that round. This definition is restricted in that it does not allow nodes in the ring to have different neighbors in different rounds. When we allow both vertex permutation and restricted 1-interval connectivity, it is equivalent to the more general version of 1-interval connectivity from \cite{OW05,KLO10} when restricted to rings. 

\subsection{Powers of Robots}
We consider $n$ robots, each having a unique label. Each robot has the ability to detect, in each round, the robot with the least label among all robots co-located with it on its node and know what that label is. For results in Sections~\ref{sec:vp},~\ref{sec:vp-1-int}, and~\ref{sec:no-vis-imp}, robots possess chirality, i.e. they share the same notion of clockwise and anticlockwise directions and can determine which edge to leave a node in order to move in the required direction. In Section~\ref{sec:no-chir-vp-1-int}, robots do not possess chirality. Finally, robots have the ability to perform weak global multiplicity detection and strong local multiplicity detection. Weak global multiplicity detection means that robots can detect if zero, one, or more than one robot is present on a given node within their visibility range. Strong local multiplicity detection means that robots can detect exactly how many robots are present on the node they are at. Robots do not require any memory for most algorithms in our paper with the exception of those presented in Sections~\ref{subsec:one-node-start} and~\ref{subsec:even-ring}.

For a given robot $u$ in a given round, the node the robot is located on in that round is denoted by $u.node$. In each round, each robot performs a look operation which allows it to update its view of the robots present in the ring. Robots have $k$-visibility, where $0 \leq k \leq n$, which constrains how many nodes away from $u.node$ a robot $u$ can see.
  
For a given direction, the distance between two nodes is defined as the number of edges between them. The view of a robot $u$ with $k$-visibility is given by the vector $<CLOCKWISE$, $ANTI-CLOCKWISE$, $MULTIPLICITY$, $MISSING-EDGE>$. $CLOCKWISE$ and $ANTI-CLOCKWISE$ are vectors whose elements are the distances between consecutive occupied nodes starting from $u.node$ and at distance at most $k$ away from $u.node$ along clockwise and anticlockwise directions respectively. $MULTIPLICITY$ is a vector whose elements are the distances of multinodes at distance at most $k$ from $u.node$ in the clockwise direction. If there are no multinodes at distance at most $k$ from $u.node$ (including the node itself), then $MULTIPLICITY$ contains just one element whose value is~$-1$. $MISSING-EDGE$ contains one element representing the distance of the closest node in the clockwise direction having a missing edge. In the case that no edge is missing, the value of $MISSING-EDGE$ is $\bot$. When $k=0$, we say that robots have \textbf{no visibility}, i.e. a robot $u$ can only observe other robots present in $u.node$. When $k \geq n/2$, we say that robots have \textbf{full visibility}, i.e. they will have information about which nodes are holes, singleton nodes, and multinodes for every node in the ring.

\subsection{Synchronicity}
Robots execute algorithms in a synchronous manner, where each round of execution is characterized by the following four stages. Initially, an adversary performs some form of dynamism in the ring. Then each robot performs a look operation, where it updates its view. Subsequently, each robot may perform some computations and detect the robot with least label among all robots co-located on the same node as it. Finally, each robot either moves to a new node or stays in its current node. Note that a given robot $u$ can only move to a node that is (i) adjacent to $u.node$ in the given round and (ii) for which the adversary has not removed the edge connecting that node and $u.node$.

\subsection{Problem Description}
Given an initial arbitrary configuration of $n$ robots present on the nodes of an $n$ node ring, the problem of dispersion requires the robots to coordinate and move such that they reach a configuration where exactly one robot is present on each node.


\section{Dispersion with Vertex Permutation}
\label{sec:vp}

The algorithm \emph{VP-Chain} works as follows. \emph{VP-Chain} is used by robots to check if they are a part of a clockwise chain in the given round. If they are, then they move to fill the hole in the clockwise chain without creating any new holes. Repeatedly performing this process allows the robots to achieve a state of dispersion, at which point no multinodes will exist and the robots will stop executing the algorithm. Thus, we arrive at the following theorem.

\alglanguage{pseudocode}
\begin{algorithm}
	\caption{VP-Chain, run by each robot $u$}
	\label{alg:vp-chain}
	\begin{algorithmic}[1]
		\While{there exists a multinode in the ring}
			\If{($u.node$ is part of a clockwise chain) and ($u$ has least label among robots in $u.node$)} 
				\State Move clockwise.
			\Else
				\State Don't move.
			\EndIf
		\EndWhile
	\end{algorithmic}
\end{algorithm}

\begin{theorem}\label{the:vp-chain-works}
	When robots have chirality and full visibility, Algorithm \emph{VP-Chain} achieves dispersion on a ring in the presence of vertex permutation dynamism in $O(n)$ rounds.
\end{theorem}
\begin{proof}

We prove that dispersion is achieved in $O(n)$ by first showing that the number of holes decreases in each round. We subsequently bound the number of possible holes in the ring.

\begin{lemma}\label{lem:vp-hole-dec}
In each round where dispersion has not yet been achieved, the number of holes in the ring decreases by at least $1$.
\end{lemma}

\begin{proof}
	In each round, if there is at least one hole, there will definitely be a multinode. Furthermore, it is clear that there will exist a clockwise chain. The algorithm directs every robot occupying a singleton node in this clockwise chain to move towards the hole. The algorithm directs exactly one robot in the multinode (the robot with the least label) to move towards the hole. Thus, for this clockwise chain, the hole is filled and no new holes are made.
	
	For every round where there exists a hole, a clockwise chain exists. Thus, in every round where dispersion hasn't yet occurred, the number of holes decreases by at least $1$.
\end{proof}

	There can be at most $n-1$ holes in the initial configuration. From Lemma~\ref{lem:vp-hole-dec}, it is clear that the upper bound on the number of rounds it takes to achieve dispersion is $n-1$ rounds. Hence \emph{VP-Chain} achieves dispersion in $O(n)$ rounds.
\end{proof}


\section{Dispersion with Vertex Permutation and 1-Interval Connectivity}
\label{sec:vp-1-int}

We now describe how algorithm \emph{VP-1-Interval-Chain} works.

If a robot is part of a multinode, then it is possible it might be part of either zero, one, or two good chains. If it's not part of a good chain, the robot doesn't move. If it's part of exactly one good chain, then the robot (if it has least label among robots on node) moves such that the hole of the chain is filled without creating any new holes. If it's part of two good chains, then the robots on the multinode coordinate to fill in the hole in the clockwise chain. 

If a robot's node is a singleton node, then it checks if it is part of a good chain. If it's not, the robot doesn't move. If it is, then the robot checks if the multinode of that chain belongs to two good chains. If so, then if the robot belongs to the anticlockwise chain, then it doesn't move, else it moves closer to the hole. 

Thus, we arrive at the following theorem.

\alglanguage{pseudocode}
\begin{algorithm}
	\caption{VP-1-Interval-Chain, run by each robot $u$}
	\label{alg:vp-1-interval-chain}
	
	\begin{algorithmic}[1]
		
		\While{there exists a multinode in the ring}
			\If{$u.node$ is singleton node and part of good chain}
				\If{the multinode of the chain is also part of another good chain that is a clockwise chain}
					\State Don't move.
						
				\Else
					\State Move towards the hole of $u.node$'s chain.
				\EndIf
				
			\ElsIf{($u.node$ is a multinode) and ($u$ has least label among robots in $u.node$)}
				\If{$u.node$ is part of two good chains}
					\State Move clockwise.
					
				\ElsIf{$u.node$ is part of one good chain}
					\State Move towards the hole of the good chain.
				\EndIf
				
			\Else
				\State Don't move.
			\EndIf
		\EndWhile
	\end{algorithmic}
\end{algorithm}

\begin{theorem}\label{the:vp-1-int-chain-works}
	When robots have chirality and full visibility, Algorithm \emph{VP-1-Interval-Chain} achieves dispersion on a ring in the presence of vertex permutation dynamism and 1-interval connectivity dynamism in $O(n)$ rounds.
\end{theorem} 
\begin{proof}
We will show that in every round prior to dispersion being achieved, there exists at least one good chain. We then argue that this implies that the number of holes in each round decreases by at least $1$. Using this, we are then able to bound the number of rounds it takes to achieve dispersion.

\begin{lemma}\label{lem:vp-1-int-good-chain-exists}
In a round where dispersion has not yet been achieved, there exists a good chain.
\end{lemma}

\begin{proof}
In a given round, if the robots have not completed dispersion, then there exists at least one hole and at least one multinode. Thus, there exist at least two chains.  As the adversary can remove at most one edge, there will be at least one good chain in the ring.
\end{proof}

\begin{lemma}\label{lem:vp-1-int-hole-dec}
In a round where dispersion has not yet been achieved, the number of holes decreases by at least $1$.
\end{lemma}

\begin{proof}
A new hole is created when all robots on an occupied node leave that node and no new robots occupy it. Robots will only move if (i) they identify that they are in a good chain and the multinode of the chain doesn't belong to any other good chain or (ii) they identify that they are in a good clockwise chain. 

These robots move such that every robot in a singleton moves towards the hole and one robot from the multinode moves towards the hole. Because no edge is missing in a good chain, the situation where robots leave one node and no robots move into that node does not occur.

Thus, when robots in a good chain move, they fill a hole. By Lemma~\ref{lem:vp-1-int-good-chain-exists}, we know that in every round until dispersion occurs there exists a good chain. Robots not in a good chain don't move and thus can't create new holes. Thus, in every round, the number of holes decreases by at least $1$.
\end{proof}

By Lemma~\ref{lem:vp-1-int-hole-dec}, we see in each round until dispersion is achieved, at least one hole is filled. Since there can be at most $n-1$ holes, it takes $O(n)$ rounds to achieve dispersion.
\end{proof}


\section{Dispersion without Chirality}
\label{sec:no-chir-vp-1-int}

\subsection{Dispersion on Rings with Restricted Initial Configurations}\label{subsec:one-node-start}
 	Let us assume that we start with a configuration where all the robots are in the same node. Algorithm \emph{No-Chir-Preprocess}  causes all robots to develop the same sense of chirality. First all robots identify a common robot, in this case the robot with the least label. This robot will serve as the key to unifying the directions among all robots. This is achieved by making every robot move ``clockwise'' and then checking if their sense of ``clockwise'' matches the least label robot's sense of ``clockwise''. All robots that are in the same node as the robot with the least label have the same directions as that of the least label robot. If a robot is not in the same node, then that robot's sense of directions will be subsequently flipped so that they match those of the least label robot's.
 	
 	Note that even in the face of vertex permutation dynamism or 1-interval connectivity dynamism, the previous ideas hold true. All that matters is that the condition that robots with the same sense of direction as the least label robot end up in the same node and those with the opposite sense of direction end up in a different node holds true. In the case of vertex permutation dynamism, the adversary permutes the ring after the robots have moved and thus the required condition holds true. In case that the edge adjacent to the initial node with all the robots is removed by the adversary, those robots that tried to move along that edge will remain in the initial node and the other robots will move to a different node. Once again, the required condition holds true.
 	
 \alglanguage{pseudocode}
 \begin{algorithm}
 	\caption{No-Chir-Preprocess, run by each robot $u$}
 	\label{alg:no-chir-preprocess}
 	\begin{algorithmic}[1]
 		\State Remember robot with least label among robots in $u.node$. Call it $x$.
 		\State Move clockwise.
 		\State If $x$ is not present in your node, reverse your clockwise and counterclockwise directions.
 	\end{algorithmic}
 \end{algorithm}

\begin{lemma}\label{lem:no-chir-to-chir}
	Algorithm \emph{No-Chir-Preprocess} ensures that robots starting on the same node will gain chirality, even in the presence of both vertex permutation dynamism and 1-interval connectivity dynamism.
\end{lemma}

\begin{proof}
	\emph{No-Chir-Preprocess} has all robots identify the robot with the least label $x$ and set their directions to match that robots sense of directions. It does this by having robots perform one round of movement and checking if $x$ moves in the same direction they did. For a given robot, this can be easily checked as if $x$ moved in the same direction, then it will be present on the same node as the robot. If $x$ and the robot did not move in the same direction, then the robot switches its sense of clockwise and anticlockwise directions. Thus, all robots will have the same sense of clockwise and anticlockwise.
\end{proof}

Due to Lemma~\ref{lem:no-chir-to-chir}, if all robots start in the same node and run \emph{No-Chir-Preprocess}, all robots will subsequently have chirality. Thus, we can then use algorithm \emph{VP-1-Interval-Chain}, developed in Section~\ref{sec:vp-1-int}, to achieve dispersion in $O(n)$ rounds in the presence of both vertex permutation dynamism and 1-interval connectivity dynamism.

\begin{theorem}\label{the:disp-wo-chir-same-node}
	When robots have full visibility, are initially present on the same node, and do not have chirality, we can achieve dispersion in $O(n)$ rounds in the presence of vertex permutation dynamism and 1-interval connectivity dynamism by running first \emph{No-Chir-Preprocess} and then \emph{VP-1-Interval-Chain}.
\end{theorem}

\begin{proof}
	We are given that robots initially start in the same node. Thus, by Lemma~\ref{lem:no-chir-to-chir}, robots develop chirality in one round by running \emph{No-Chir-Preprocess}. By Theorem~\ref{the:vp-1-int-chain-works}, robots will then achieve dispersion in $O(n)$ rounds by running \emph{VP-1-Interval-Chain}. Thus, all robots achieve dispersion in $O(n)$ rounds.
\end{proof}

\subsection{Dispersion on Rings with Odd Number of Nodes}\label{subsec:odd-ring}
If the ring has an odd number of nodes, then even with an arbitrary initial distribution of robots to nodes, we are able to achieve dispersion. Robots running \emph{Achiral-Odd-VP-1-Interval-Chain} achieve dispersion in $O(n)$ rounds even in the presence of both vertex permutation and 1-interval connectivity dynamism. 

This algorithm is similar to the algorithm \emph{VP-1-Interval-Chain} in Section~\ref{sec:vp-1-int} with two key differences. The first difference is that if a robot is on a multinode that is part of two good chains and must move, tie-breaking is done first based on which chain is shorter in length and only then if needed on direction. The second difference is that if a robot is on a singleton node in a good chain, it checks if the multinode of the chain it belongs to belongs to another shorter or equal length good chain before deciding to move. 
These differences are the result of robots relying on chain size instead of a common sense of direction to unify their movements.

The result achieved by the algorithm is characterized by the following theorem.

\alglanguage{pseudocode}
\begin{algorithm}
	\caption{Achiral-Odd-VP-1-Interval-Chain, run by each robot $u$}
	\label{alg:achiral-odd-vp-1-interval-chain}
	
	\begin{algorithmic}[1]
		\While{there exists a multinode in the ring}
			\If{$u.node$ is singleton node and part of good chain}
				\If{multinode of the chain is also part of another good chain that is shorter or same in length}
					\State Don't move.
				\Else
					\State Move towards the hole of $u.node$'s chain.
			\EndIf

			\ElsIf{($u.node$ is a multinode) and ($u$ has least label among robots in $u.node$)}
				\If{$u.node$ is a part of two good chains}
					\If{both chains have same length}
						\State Move clockwise.
					\Else
						\State Move towards the hole of the shorter chain.
					\EndIf 
				\ElsIf{$u.node$ is a part of one good chain}
					\State Move towards the hole of the good chain.
				\EndIf

			\Else
				\State Don't move.
			\EndIf
		\EndWhile				
	\end{algorithmic}
\end{algorithm}

\begin{theorem}\label{the:achiral-odd-vp-1-interval-chain-works}
	When robots have full visibility and do not have chirality, Algorithm \emph{Achiral-Odd-VP-1-Interval-Chain} achieves dispersion on a ring of odd length in the presence of vertex permutation dynamism and 1-interval connectivity dynamism in $O(n)$ rounds.
\end{theorem}

\begin{proof}
We first argue that the number of holes will eventually decrease until none are left and subsequently bound the time this will take.

We argue the eventual decrease in holes as follows. We show, by Lemma~\ref{lem:hole-dec-multi-inc}, that in each round, either (i) the number of holes decreases or (ii) the number of holes remains constant and the number of multinodes increases. Since there is an upper bound of $n/2$ to the number of multinodes that can exist in an $n$ node ring, eventually it will become impossible for the number of multinodes to increase. Thus, the number of holes can only decrease from that round onwards and eventually dispersion will be achieved.

\begin{lemma}\label{lem:hole-dec-multi-inc}
	In each round, either (i) the number of holes decreases or (ii) the number of holes remains constant and the number of multinodes increases.
\end{lemma}
\begin{proof}
	We show that if the lemma does not hold, it implies that the ring is of even size. Because the ring is not of even size and by the contrapositive of the previous statement, our lemma holds true. We first enumerate the different types of movements possible by robots and subsequently show that only one type of movement will lead to the lemma not holding.
	
	We know that there will exist at least one multinode and one hole if dispersion has not occurred and hence at least one good chain. From the algorithm, it is seen that the least label robot $u$ of a multinode that is part of at least one good chain executes a move in every round. This robot either moves towards the hole of one of its good chains or moves clockwise. 
	
	$u$ moves towards the hole of a given good chain $c$ if either $u.node$ is part of exactly one good chain $c$ or $u.node$ is part of two good chains and $c$ is shorter in length. In either case, all robots in singleton nodes, if any, in $c$ move towards the hole too and hence a hole gets occupied without creating any new holes. This decreases the number of holes. 
	
	If $u.node$ belongs to two good chains of equal length, then $u$ moves clockwise. However, the robots of the singleton nodes, if any, of both good chains remain stationary. This means that $u$ moves into its neighboring node which could be a hole or a singleton node. If it is a hole, it gets filled and hence the number of holes decreases. If it is a singleton, it becomes a multinode. If the multinode $u$ was initially on had more than $2$ robots, then this movement of $u$ leads to an increase in the number of multinodes, else the number of multinodes remains the same. 
	
	The lemma does not hold if either (i) the number of holes increases or (ii) the number of holes does not decrease and the number of multinodes does not increase. Condition (i) can never occur as robots move only if their nodes are part of a good chain. When the robots on the nodes of a good chain move, it is not possible to create a hole. Thus, we only consider condition (ii). Now, the only configurations where condition (ii) occurs in a round is when all the multinodes belong to two equal length good chains and have multiplicity exactly two. This follows from our previous enumeration of the different types of movements of robots possible. As the number of holes doesn't decrease, each multinode should belong to two equal length good chains. And as the number of multinodes remains constant, each multinode should have exactly two robots. 
	
	Let us assume that there are $k$ multinodes initially. Since each multinode has exactly two robots, there are $k$ holes in the ring. Each multinode has two choices of holes to fill. Hence $k$ multinodes on a whole have $2k$ choices but there are only $k$ holes in ring. This condition can only be satisfied if each hole is shared by two multinodes. Hence there needs to be an alternating pattern of multinodes and holes, i.e. every pair of holes must be separated by at least one multinode and every pair of multinodes must be separated by at least one hole. Since each multinode must belong to two equal length chains, the number of singleton nodes $s$ must be even due to this alternating pattern of multinodes and holes. Thus, the total number of nodes of the ring $s + 2k$ is an even number. However, we assume that the ring is of odd size and thus the lemma will hold.
\end{proof}

We now bound the running time of the algorithm. In each round, either the number of multinodes increases and number of holes remains constant, or the number of holes decreases. Let $rds\_multi\_inc$ be the number of rounds that the number of multinodes can increase while the number of holes remains constant. Let $rds\_hole\_dec$ be the number of rounds that the number of holes decreases. The running time of the algorithm is $\leq rds\_multi\_inc+rds\_hole\_dec$. We bound each of these terms to get an asymptotic value for the running time.

We can easily see that $rds\_hole\_dec \leq n-1$. We show that $rds\_multi\_inc$ is upper bounded by two values, which we bound in turn. Let $diff\_initial\_max\_multi$ be the difference between the number of multinodes in the initial configuration and the maximum number of multinodes possible. Let $total\_multi\_dec$ be the total decrease in the number of multinodes over the course of the algorithm. The number of multinodes can only increase up to a certain maximum value. However, every time the number of multinodes decreases, they can subsequently increase in the future. Thus, it is clear that $rds\_multi\_inc \leq diff\_initial\_max\_multi + total\_multi\_dec$. We now individually bound these two terms. 

The total number of multinodes that can exist in an $n$ node ring is $n/2$. Thus $diff\_initial\_max\_multi \leq n/2$. We now focus on bounding $total\_multi\_dec$. When the number of holes decreases, the number of multinodes can either increase, decrease, or stay the same. We show, by Lemma~\ref{lem:hole-occupy-multi-dec}, if in a given round the number of holes decreases by some number $x$, the number of multinodes can also decrease by at most $x$. Therefore, over the course of the algorithm, the maximum decrease in the number of multinodes is $\leq n-1$. Thus, $total\_multi\_dec \leq n-1$.

Therefore, the running time of the algorithm is 
\begin{align*}
&\leq rds\_multi\_inc+rds\_hole\_dec \\
&\leq diff\_initial\_max\_multi + total\_multi\_dec+rds\_hole\_dec \\
&\leq n/2 + n-1 + n-1\\
&= O(n) \text{ rounds.}
\end{align*} 

\begin{lemma}\label{lem:hole-occupy-multi-dec}
	In each round, for every hole that becomes occupied, there can be at most a decrease of one multinode.
\end{lemma}
\begin{proof}
	Any hole will be a part of at most two good chains since the graph is a ring. It can become occupied by either receiving a robot from exactly one good chain or from both good chains. If it receives a robot from exactly one good chain, it is possible for only the multinode in that chain to become a singleton node. If it receives a robot from both good chains, then the hole itself becomes a multinode and at most both multinodes of the chains become singletons. Thus, for every hole that becomes occupied in a given round, the number of multinodes in the graph decreases by at most one.
\end{proof}

 Thus the theorem is proved true.
\end{proof}

\subsection{Dispersion on a Ring of Size $4$}\label{subsec:even-ring}
We now present algorithm \emph{Achiral-Even4-VP-1-Interval-Chain} to achieve dispersion on a ring of size $4$ when robots start in an arbitrary initial configuration. This algorithm bears a strong similarity to \emph{Achiral-Odd-VP-1-Interval-Chain} as in most configurations the algorithms work similarly. 

However, the key difference is the handling of the configuration where a multinode has two singleton nodes adjacent to it, and the edges between the multinode and both singleton nodes are present in the round. This situation, when presented to the algorithm \emph{Achiral-Odd-VP-1-Interval-Chain}, would result in a configuration that can be manipulated by the adversary to again reach this situation. In other words, the adversary could then force the robots to run the algorithm forever without ever achieving dispersion. We handle this particular situation by having robots in the singleton nodes move to the multinode and then subsequently having all robots run \emph{No-Chir-Preprocess} and \emph{VP-1-Interval-Chain}.

Thus we are able to achieve dispersion in $O(1)$ rounds on a ring of size $4$, as captured by the following theorem.

\alglanguage{pseudocode}
\begin{algorithm}
	\caption{Achiral-Even4-VP-1-Interval-Chain, run by each robot $u$}
	\label{alg:achiral-even4-vp-1-interval-chain}
	\begin{algorithmic}[1]
		\While{there exists a multinode in the ring}
			\If{$u.node$ has all $4$ robots on it}
				\State \emph{No-Chir-Preprocess}
				\State \emph{VP-1-Interval-Chain}
				
			\ElsIf{$u.node$ is singleton node}
				\If{$u.node$ is part of a good chain}
					\If{the multinode of the chain is also part of another good chain that is shorter}
						\State Don't move.
					\ElsIf{multinode of the chain is also part of another good chain that is same in length}
						\State Move towards the multinode.
					\Else
						\State Move towards the hole of $u.node$'s chain.
					\EndIf
				\EndIf
		
			\ElsIf{($u.node$ is a multinode) and ($u$ has least label among robots in $u.node$)}
				\If{$u.node$ is a part of two good chains}
					\If{chains have different length}
						\State Move towards the hole of the shorter chain.
					\ElsIf{$u.node$ is adjacent to two holes}
						\State Move clockwise. 
					\EndIf 
				\ElsIf{$u.node$ is a part of one good chain}
					\State Move towards the hole of the good chain.
				\EndIf
		
			\Else
				\State Don't move.
			\EndIf
		\EndWhile					
	\end{algorithmic}
\end{algorithm}

\begin{theorem}\label{the:achiral-even4-vp-1-interval-chain}
	When robots have full visibility and do not have chirality, Algorithm \emph{Achiral-Even4-VP-1-Interval-Chain} achieves dispersion on a ring of size $4$ in the presence of vertex permutation dynamism and 1-interval connectivity dynamism in $O(1)$ rounds.	
\end{theorem}

\begin{proof}
	Recall that we are dealing with a ring of size $4$ and by extension, only $4$ robots. Let us characterize the state of the system by the number of multinodes and singleton nodes present in the system. Any system where dispersion has not yet been achieved will be in one of the following four states: 
	\begin{enumerate}
		\item There is one multinode with $4$ robots on it.
		\item There is one multinode with $3$ robots on it and one singleton node.
		\item There is one multinode with $2$ robots on it and two singleton nodes.
		\item There are two multinodes with $2$ robots on each.
	\end{enumerate} 
	
	We prove that dispersion is achieved in $O(1)$ rounds from any of these four states using \emph{Achiral-Even4-VP-1-Interval-Chain}. Note that our goal is dispersion and so a system already in that state implies our algorithm achieves dispersion in $O(1)$ rounds. 
	
	It is clear to see that when the system is in state 1, the condition in Line~2 of the algorithm is satisfied and \emph{No-Chir-Preprocess} and \emph{VP-1-Interval-Chain} will be executed. By Theorems~\ref{the:disp-wo-chir-same-node}, and given $n=4$, robots in state 1 achieve dispersion in $O(1)$ rounds.
	
	For states 2-4, we first show, in Lemma~\ref{lem:states24-to-3}, that any system in state 2 or 4 changes to state 3 or a state of dispersion in one round. We then show, in Lemma~\ref{lem:state3-disp}, that a system in state 3 achieves dispersion in $O(1)$ rounds.
	
	\begin{lemma}\label{lem:states24-to-3}
		If the configuration of robots on nodes is characterized by state 2 or 4, then after robots run \emph{Achiral-Even4-VP-1-Interval-Chain} for one round, the system will reach state 3 or a state of dispersion.
	\end{lemma}
	
	\begin{proof}
		We consider each state in turn, and show that in one round of executing \emph{Achiral-Even4-VP-1-Interval-Chain}, we arrive at state 3 or a state of dispersion.
		
		Consider a system in state 2. There is one multinode and one singleton node. Either these two nodes are (i) opposite to each other or (ii) adjacent to each other. 
		
		If the multinode and singleton node are opposite, then only the least label robot in the multinode will move in the round. Regardless of if the adversary makes one of the chains of the multinode a bad chain, the least label robot will move to the hole in one of the chains and state 3 will be reached in one round.
		
		If the multinode and singleton node are adjacent to each other, then depending on which edge the adversary removes, if any, there will either be exactly one good chain or there will be a smaller good chain and a bigger good chain. In both cases, exactly one hole will be filled in the round and the system will change to state 3.
		
		Consider a system in state 4. Either the two multinodes are (i) opposite to each other or (ii) adjacent to each other.
		
		If the multinodes are opposite to each other, then depending on if the adversary removes an edge or not, the condition of either Line~17 or Line~19 is satisfied for exactly one robot on each multinode. Now either dispersion is achieved or the system changes to state 3 in one round.
		
		If the multinodes are adjacent to each other, then depending on which edge the adversary removes, if any, either one or both multinodes will be a part of a good chain. If only one multinode belongs to a good chain, then state 3 is reached in one round. If both multinodes belong to good chains, then dispersion is achieved in one round.
		
		Thus it is clear that, whether the system is in state 2 or 4, in one round, the system will change to state 3 or a state of dispersion.		
	\end{proof}
	
	Now that we've shown that systems in state 2 or 4 change to state 3 or a state of dispersion in one round, all that's left to show is that we can achieve dispersion in $O(1)$ rounds once we reach state 3.
	
	\begin{lemma}\label{lem:state3-disp}
		If the configuration of robots on nodes is characterized by state 3, then after robots run \emph{Achiral-Even4-VP-1-Interval-Chain} for $O(1)$ rounds, they will achieve dispersion.
	\end{lemma}
	
	\begin{proof}
		We show that state 3 reaches either a state of dispersion or state 1 in one round of executing \emph{Achiral-Even4-VP-1-Interval-Chain}. When the system is in state 3, either (i) the multinode will either be adjacent to both singleton nodes or (ii) the multinode will be adjacent to one singleton node and a hole.
		
		Consider case (i). Now, the adversary can either cause one of the chains the multinode belongs to be a bad chain or allow both chains to be good. If both chains are good chains, the robots in the multinode won't move and the robots in the singleton nodes will move towards the multinode, thus bringing the system to state 1 in one round. If only one chain is a good chain, then the robot on the singleton of the bad chain remains stationary while the least label robot of the multinode and the robot on the singleton node of the good chain move toward the hole. Thus a state of dispersion is reached in one round.
		
		Consider case (ii). Now, the adversary can either cause one of the chains the multinode belongs to be be a bad chain or allow both chains to be good. If both chains are good chains, then since one chain is shorter than the other, the least label robot in the multinode moves to the hole and a state of dispersion is reached in one round. If only one chain is a good chain, then still, the least label robot of the multinode and the robots in the singleton nodes in the good chain, if any, move toward the hole and a state of dispersion is reached in one round.
		
		Thus we see that robots in a configuration characterized by state 3 reach either a state of dispersion or state 1 in one round of executing \emph{Achiral-Even4-VP-1-Interval-Chain}. As we saw earlier, once the system is in state 1, robots will achieve dispersion in $O(1)$ rounds. Thus, the lemma is proved.
	\end{proof}
	
	Thus, the robots on a ring of size $4$ will achieve dispersion in $O(1)$ rounds, regardless of the starting configuration of the robots on the nodes.
\end{proof}


\section{Impossibility of Dispersion with No Visibility}
\label{sec:no-vis-imp}

In this section we look at the impossibility of dispersion in the face of the two types of dynamism when robots have no visibility. Note that these impossibility results hold even if the robots possess chirality. Furthermore, we even allow the robots a little more power than our model for the possibility results in previous sections. We allow robots co-located on a node to detect not just the robot with the least label among them, but also the robot with the second least label among them. Even with this added power, dispersion is still impossible. The key idea to the impossibility results is that the adversary is able to predict the robots' movements in a given round before they move and so modify the ring to prevent dispersion from being achieved. First, we note the following observation about the dependence of a robot's move on only local information.

\begin{observation}\label{obs:local-info}
If robots have no visibility, then in a given round, for every robot $u$, its choice of move depends only on information present in $u.node$.
\end{observation}

Observation~\ref{obs:local-info} implies that any robot $u$'s move in a given round does not depend on information from nodes other than $u.node$. Thus, the adversary may change the ring, either by permuting the nodes or removing an edge, and $u$'s choice of movement wouldn't be affected. This allows the adversary to predict how the robots will move and prevent dispersion from occurring by ensuring the presence of either a multinode or a hole in every round.

\subsection{With Vertex Permutation}
\begin{theorem}\label{the:no-vis-vp-imp}
	If robots have no visibility, it is impossible to achieve deterministic dispersion in a ring with $\geq 3$ nodes in the presence of vertex permutation dynamism.
\end{theorem}

\begin{proof}
	We argue the impossibility result as follows. Recall that the connections in the graph change in every round. When the adversary permutes the ring, only the positions of robots relative to other robots in the graph changes, but there is no change to robots co-located on the same node. We first note that for a given robot $u$, by Observation~\ref{obs:local-info}, its choice of move depends only on information accessible to it on $u.node$. Thus, the adversary knows each robot's plan to move before permuting the ring. The adversary leverages this information to ensure the presence of at least one multinode in every round, which prevents dispersion of robots. Due to a slight difference in possible configurations of the ring to avoid, we divide the proof into two parts: for those rings of size $3$ and those of size $\geq 4$.
	
	\begin{lemma}\label{lemma::no-vis-vp-imp-n3}
		If robots have no visibility, it is impossible to achieve deterministic dispersion in a ring with $3$ nodes in the presence of vertex permutation dynamism.
	\end{lemma}
	
	\begin{proof}
		For rings of size 3, it is not enough to ensure that a multinode exists in each round. We must also ensure that all 3 robots do not end up on the same node, henceforth called a \textbf{gathering configuration}, else the adversary will be unable to prevent dispersion from occurring. This is because once such a configuration is reached, all robots are present in the same node and permuting the nodes does not change that. After the vertex permutation, the robots can coordinate such that one robot goes through each edge and one robot stays in the node, thus achieving dispersion.
		
		Thus, the adversary should ensure that only ring configurations where there exists one hole, one multinode and one singleton node occur in every round. Let us call all such configurations \textbf{neutral configurations}. Let us assume the ring is initially in one such neutral configuration. We now show that regardless of how the robots plan to move, the adversary has a way to permute the ring so that robots are always in a neutral configuration.
		
		If the movement of robots will lead to a neutral configuration, the adversary does not need to permute the ring. 
		
		In case the movement of robots will lead to a gathering configuration, the adversary performs the following operations.
		There are three possible ways to achieve a gathering configuration: (i) the robot from the singleton node moves towards the multinode, (ii) both robots of the multinode move towards the singleton (iii) all robots move towards the hole. Cases (i) and (ii) are similar and require the adversary to swap the multinode and the singleton node. This results in the robot(s) which performs the move operation to move towards the hole and a neutral configuration being maintained. In case (iii), the adversary swaps the multinode and singleton. This results in the robot of the singleton node moving towards the multinode and the robots of the multinode moving towards the singleton. Hence, a neutral configuration is still maintained.
		
		In case the movement of robots will lead to dispersion, the adversary performs the following operations.
		In order to achieve dispersion one robot should move into the hole and there should be one robot less in the multinode at the end of the round. This is possible if (i) one robot moves out of the multinode while the other stays in it and the robot in the singleton node does not move or (ii) both robots in the multinode move out of it and the robot in the singleton node moves into the multinode. 
		In case (i), the adversary can stay in a neutral configuration by swapping the multinode and the hole. The robot leaving the multinode will move into the singleton node and make it a multinode with exactly 2 robots in it. 
		In case (ii), the adversary can stay in a neutral configuration by swapping the singleton node and the hole. This results in the multinode becoming a hole, the singleton node becoming a multinode, and the hole becoming a singleton.
	\end{proof}
	
	\begin{lemma}\label{lemma::no-vis-vp-imp-n4more}
		If robots have no visibility, it is impossible to achieve deterministic dispersion in a ring with $\geq 4$ nodes in the presence of vertex permutation dynamism.
	\end{lemma}
	
	\begin{proof}
	All ring configurations where dispersion hasn't occurred can be reduced to three types, based on the number of holes in the ring. The three cases are when the number of holes are $1$, $2$, and $\geq 3$. In each of these three cases, we show that if the robots can move given the current ring configuration and achieve dispersion, the adversary has a way to permute the nodes such that dispersion won't occur. We do this by showing that the adversary can always maintain a multinode or a hole.
	
	Note that if the movement of robots does not lead to dispersion, the adversary can just maintain the ring as it is. Only in case the movement might lead to dispersion should the adversary perform the following operations.\\
	
	\textbf{For $1$ hole case:}
	
	When there is exactly $1$ hole in the network, there will be exactly one multinode with exactly $2$ robots on it. In order to achieve dispersion in a given round from this configuration, robots should move such that the following two conditions are satisfied: (i) exactly one robot moves into the hole from a neighboring node of the hole and (ii) the total number of robots present on the multinode at the end of the round should be exactly $1$. 
	
	\begin{figure}
	\includegraphics[page=1,height=2in]{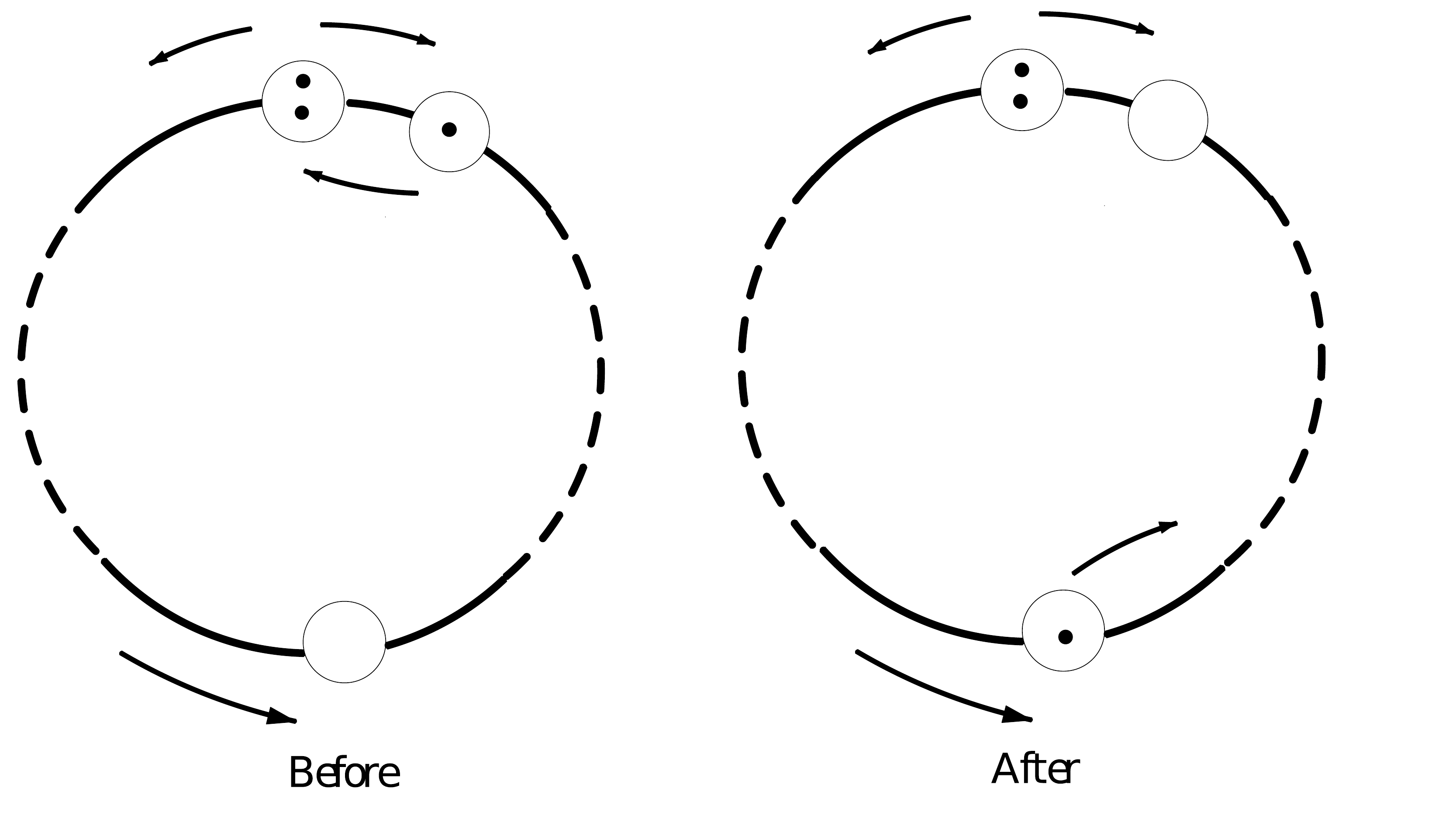}
	\caption{Adversary swapping nodes to address case (ii).a of 1-hole impossibility.}\label{fig:1-hole-ii.a} 
	\end{figure}
	
	\begin{figure}
		\includegraphics[page=2,height=2in]{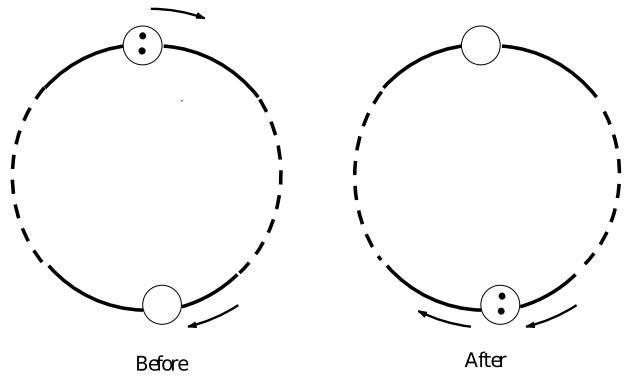}
		\caption{Adversary swapping nodes to address case (ii).b of 1-hole impossibility.}\label{fig:1-hole-ii.b} 
	\end{figure}
	
	Condition (ii) can be satisfied in one of two ways. Either (ii.a) both robots in the multinode move out of the node and a robot in a neighboring node moves into the node or (ii.b) one robot stays in node, the other robot moves to a neighboring node, and no robot from a neighboring node moves into the multinode. If the planned robot movements satisfy condition (ii.a), then the adversary swaps the hole with the neighboring node specified in the condition, as illustrated in Figure~\ref{fig:1-hole-ii.a}. No robot will move into the multinode and it will become a hole. Hence dispersion will not be achieved. If the planned robot movements satisfy condition (ii.b), then the adversary swaps the multinode and the hole, as illustrated in Figure~\ref{fig:1-hole-ii.b}. Since no robot moves into the hole, it will remain a hole and dispersion will not be achieved.\\
	
	\textbf{For $2$ holes case:}

	Two holes present in the system arise when either (i) a multinode has three robots or (ii) two multinodes have two robots each. We look at each case in turn.\\
	
	\emph{(i) A multinode has three robots:}
	
	In order to achieve dispersion in a given round the following two conditions should be satisfied: (1) two robots on the multinode should be move in different directions while one stays on the node and (2) there should be one incoming robot into each of the holes. The reason for this specific movement of robots on the multinode is that if more than one robot moves in either direction or stays on the node, a multinode will definitely exist at the end of the round and dispersion will not be achieved. Condition (2) can be satisfied in one of two ways: (2.a) at least one of the robots that move into the holes comes from a singleton node or (2.b) both holes receive a robot from the multinode.
	
	\begin{figure}
		\includegraphics[page=7,height=2in]{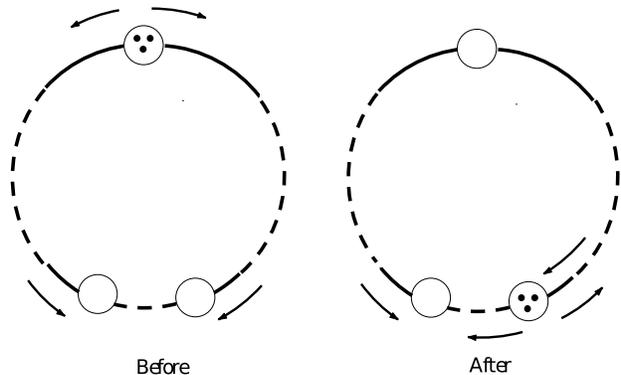}
		\caption{Adversary swapping nodes to address case 2.a of 2-hole impossibility when one multinode has three robots.}\label{fig:2-hole-3-rob-2.a} 
	\end{figure}
	
	\begin{figure}
		\includegraphics[page=8,height=2in]{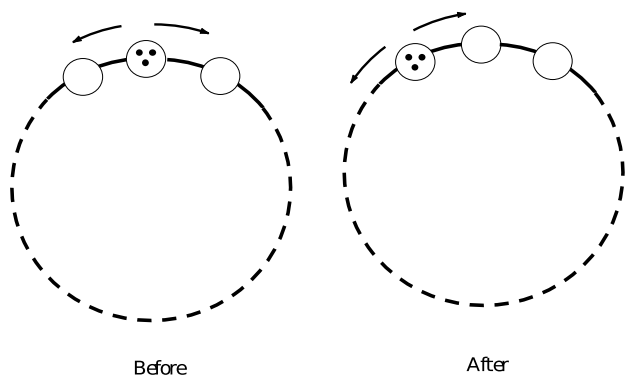}
		\caption{Adversary swapping nodes to address case 2.b of 2-hole impossibility when one multinode has three robots.}\label{fig:2-hole-3-rob-2.b} 
	\end{figure}
	
	If robot movements satisfy condition (2.a), the adversary swaps  the multinode with a hole that would have received a robot from a singleton node, as illustrated in Figure~\ref{fig:2-hole-3-rob-2.a}. If both holes would have received a robot from a singleton node, the adversary choose one arbitrarily. After the swap operation, the robot from the singleton node will move into the multinode. And since there will already be a robot in the multinode after robots in it move in order to satisfy condition (1), it still remains a multinode. Thus, dispersion is not achieved at the end of the round.
	
	If robot movements satisfy condition (2.b), the adversary swaps the multinode with any one of the holes, as illustrated in Figure~\ref{fig:2-hole-3-rob-2.b}. The hole which was not swapped does not have an incoming robot from any other node and remains a hole. Thus, dispersion is not achieved at the end of the round.\\
	
	\emph{(ii) Two multinodes have two robots each:}
		
	In order to achieve dispersion in a given round, the following two conditions should be satisfied: (1) for each hole, there should be one robot that will move into it and (2) for each multinode, either the robots move in opposite directions or one robot moves and the other stays on the node. The reason for the specific movements in condition (2) is that if more both robots decide to move in the same direction or stay, then the adjacent node in that direction will become a multinode or the original multinode will stay a multinode respectively. Let the two multinodes be $m_1$ and $m_2$. There are four ways condition (2) may be satisfied: (2.a) for both $m_1$ and $m_2$, the robots move in opposite directions, (2.b) for both $m_1$ and $m_2$, one robot moves and the other stays on the node, (2.c) for $m_1$, one robot moves and one stays and for $m_2$, both robots move in opposite directions or (2.d) for $m_1$, both robots move in opposite directions and for $m_2$, one robot moves and one stays.
	
	\begin{figure}
		\includegraphics[page=3,height=2in]{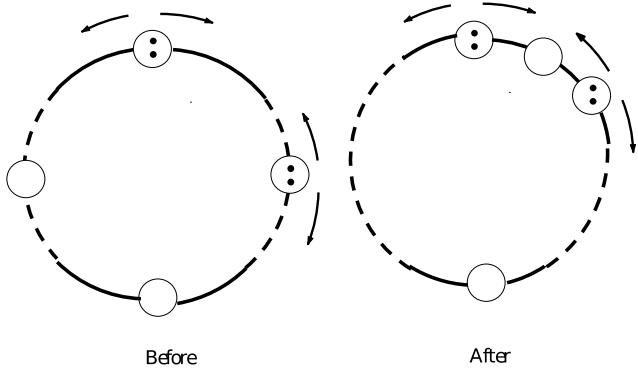}
		\caption{Adversary swapping nodes to address case 2.a of 2-hole impossibility when two multinodes have two robots each.}\label{fig:2-hole-2-multi-2.a} 
	\end{figure}
	
	\begin{figure}
		\includegraphics[page=4,height=2in]{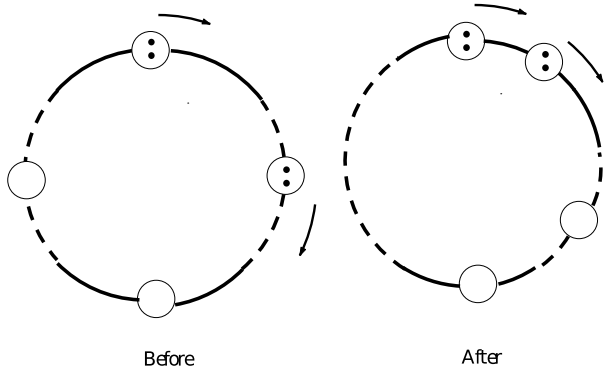}
		\caption{Adversary swapping nodes to address case 2.b of 2-hole impossibility when two multinodes have two robots each. Both robots that choose to move, move in the same direction.}\label{fig:2-hole-2-multi-2.bi} 
	\end{figure}
	
	\begin{figure}
		\includegraphics[page=5,height=2in]{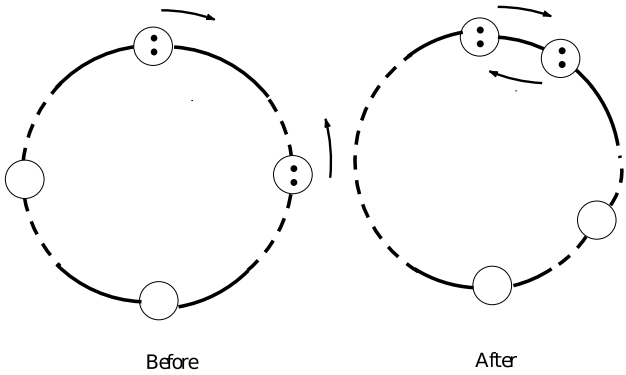}
		\caption{Adversary swapping nodes to address case 2.b of 2-hole impossibility when two multinodes have two robots each. The robots that choose to move, move in opposite directions.}\label{fig:2-hole-2-multi-2.bii} 
	\end{figure}
	
	\begin{figure}
		\includegraphics[page=6,height=2in]{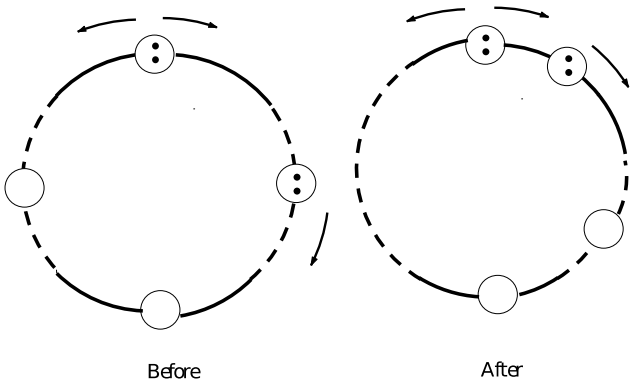}
		\caption{Adversary swapping nodes to address cases 2.c \& 2.d of 2-hole impossibility when two multinodes have two robots.}\label{fig:2-hole-2-multi-2.c-2.d} 
	\end{figure}
	
	If robot movements satisfy condition (2.a), the adversary can permute nodes so that the three node sequence consisting of $m_1$, a hole, and $m_2$ in that clockwise order arises, as illustrated in Figure~\ref{fig:2-hole-2-multi-2.a}. Since the robots of both $m_1$ and $m_2$ move in opposite directions, the hole which is adjacent to both of them becomes a multinode. Thus, dispersion is not achieved at the end of the round.
	
	If robot movements satisfy condition (2.b), the adversary can make $m_1$ and $m_2$ adjacent to each other. The exact order of $m_1$ and $m_2$ depends on the relative directions of the moving robots. If the moving robots move in same direction, it doesn't matter since either $m_1$ will have a robot enter it from $m_2$ or vice-versa, as illustrated in Figure~\ref{fig:2-hole-2-multi-2.bi}. Since $m_1$ and $m_2$ both retain one robot, the addition of this robot ensures a multinode exists. If the moving robots move in opposite directions, then we place $m_1$ and $m_2$ next to each other such that the moving robot from $m_1$ moves into $m_2$ and vice-versa, as illustrated in Figure~\ref{fig:2-hole-2-multi-2.bi}. This ensures that both nodes remain multinodes. Thus, dispersion is not achieved at the end of the round.
	
	If robot movements satisfy either condition (2.c) or (2.d), the adversary can move $m_1$ and $m_2$ so that they are adjacent to each other, as illustrated in Figure~\ref{fig:2-hole-2-multi-2.c-2.d}. Since one of the nodes retains a robot and a robot will move into that node from the other node, a multinode will exist. Thus, dispersion is not achieved at the end of the round.\\
	
	\textbf{For $\geq 3$ holes case:}
	
	In this case the adversary can permute the nodes in such a way that all the holes are consecutive. This ensures that there exists a hole with two adjacent holes. This hole can never be filled by any movement of the robots. 
	\end{proof}

	Putting together Lemmas~\ref{lemma::no-vis-vp-imp-n3} and~\ref{lemma::no-vis-vp-imp-n4more}, we see that the theorem is true.
\end{proof}

Note that this proof hinges heavily on Observation~\ref{obs:local-info} that the adversary is able to predict the movement of the robots in a given round prior to them moving. In the case where robots are allowed to use randomization to determine movements, if we give the adversary extra powers to see those bits, then dispersion even with the allowance of randomness is impossible. This is captured by the following corollary.

\begin{corollary}\label{cor:vp-stronger-adv}
	At the beginning of any round, if the adversary is allowed to know the random bits, if any, that robots will use to determine movement in that round, then the following holds. If robots have no visibility, it is impossible to achieve dispersion in a ring with $\geq 3$ nodes in the presence of vertex permutation dynamism.
\end{corollary}

\subsection{With 1-Interval Connectivity}

\begin{theorem}\label{the:no-vis-vp-1-int-imp}
	If robots have no visibility, it is impossible to achieve deterministic dispersion in a ring with $\geq 2$ nodes in the presence of 1-interval connectivity dynamism.
\end{theorem}

\begin{proof}
	From Observation~\ref{obs:local-info}, the adversary can predict the movement of the robots in a given round at the beginning of that round. We show that in any round, regardless of the movements of the robots, the adversary can remove at most a single edge to prevent dispersion from being achieved in that round. Thus, in all rounds dispersion will be prevented from being achieved. 
	
	In a given round, if the movement of robots does not lead to dispersion, the adversary need not remove any edge. 
	
	We need to consider only the movement of robots that leads to dispersion at the end of a given round. In order for dispersion to occur, all holes in the ring must have an incoming robot from exactly one of their neighboring nodes. This is because if robots came from both the neighboring nodes of the hole, the hole will become a multinode and dispersion will not occur. 
	
	Consider one such hole and the neighboring node from which a robot will move to the hole. The adversary can just remove the edge connecting the hole and that neighboring node. Thus, at the end of that round, the hole remains and dispersion does not occur.
\end{proof}

Similar to Corollary~\ref{cor:vp-stronger-adv}, when we have a stronger adversary that can predict the moves the robots will make in the current round, even when randomness is involved, dispersion is impossible in the face of 1-interval connectivity.

\begin{corollary}\label{cor:vp-1-int-stronger-adv}
	At the beginning of any round, if the adversary is allowed to know the random bits, if any, that robots will use to determine movement in that round, then the following holds. If robots have no visibility, it is impossible to achieve dispersion in a ring with $\geq 2$ nodes in the presence of 1-interval connectivity dynamism.
\end{corollary}


\section{Conclusions and Future Work}

\label{sec:conc}

In this work, we have developed asymptotically time optimal algorithms to achieve dispersion when robots have full visibility, with and without chirality, in the face of vertex permutation dynamism and 1-interval connectivity dynamism. We have also proved the impossibility of achieving dispersion when robots have no visibility. Recall that any solution to dispersion acts as a solution to $n$ robot collaborative exploration and $n$ robot scattering under the same assumptions and model. Thus all our algorithms also solve the problems of $n$ robot collaborative exploration and $n$ robot scattering on a ring of size $n$ for the given model and constraints of the algorithms.

Two interesting open questions arise from this work.\\
\textbf{Open Question 1:} What can be said of the possibility or impossibility of achieving dispersion when visibility is limited to more than no visibility but less than full visibility?\\ 
\textbf{Open Question 2:} Is dispersion of achiral robots on an even node ring of size $\geq 6$, in the presence of both vertex permutation and 1-interval connectivity dynamism, possible or impossible?


\section*{Acknowledgements}
We thank the anonymous reviewers of a previous version of this paper for pointing us to several relevant references and for other comments that improved the presentation of this paper.

\bibliographystyle{abbrv}
\bibliography{references} 

\end{document}